\documentclass{llncs}
\usepackage{amsmath}
\usepackage{amssymb}
\usepackage{tikz}
\usepackage{graphicx,color}
\usepackage{multicol}
\usepackage{bussproofs}
\usepackage{float}
\usepackage{centernot}
\usepackage{marginnote}
\usepackage{hyperref}

\usetikzlibrary{arrows,shapes,calc}

\DeclareFontFamily{U}{wncy}{}
\DeclareFontShape{U}{wncy}{m}{n}{<->wncyss10}{}
\DeclareSymbolFont{mcy}{U}{wncy}{m}{n}
\DeclareMathSymbol{\sh}{\mathord}{mcy}{"78}

\newcommand{\wlp}{\mathsf{wlp}}

\newcommand{\test}{\mathsf{test}}
\newcommand{\eval}{\mathsf{eval}}
\newcommand{\Id}{\mathsf{Id}}
\newcommand{\edn}{\mathsf{end}}
\newcommand{\unchanged}{\mathsf{unchanged}}
\newcommand{\preserves}{\mathsf{preserves}}

\newtheorem{thm}{Theorem}

\begin{document}

\tikzstyle{elem} = [circle]
\tikzstyle{line} = [draw,thick, -latex']
\tikzstyle{rel} = [draw,thin,dashed, -latex']

\title{Algebraic Principles for Rely-Guarantee Style Concurrency Verification Tools}

\author{Alasdair Armstrong \and Victor B.~F.~Gomes \and Georg Struth}

\institute{Department of Computer Science, University of Sheffield, UK\\
\email{$\{$a.armstrong,v.gomes,g.struth$\}$@dcs.shef.ac.uk}}

\maketitle

\begin{abstract}
  We provide simple equational principles for deriving
  rely-guarantee-style inference rules and refinement laws based on
  idempotent semirings. We link the algebraic layer with concrete
  models of programs based on languages and execution traces. We have
  implemented the approach in Isabelle/HOL as a lightweight
  concurrency verification tool that supports reasoning about the
  control and data flow of concurrent programs with shared variables
  at different levels of abstraction. This is illustrated on two
  simple verification examples.
\end{abstract}

\pagestyle{plain}

\section{Introduction}

Extensions of Hoare logics are becoming increasingly important for the
verification and development of concurrent and multiprocessor
programs. One of the most popular extensions is Jones' rely-guarantee
method~\cite{jones_development_1981}.  A main benefit of this method is
compositionality: the verification of large concurrent programs can be
reduced to the independent verification of individual subprograms. The
effect of interactions or interference between subprograms is captured
by \emph{rely} and \emph{guarantee} conditions. Rely conditions describe
the effect of the environment on an individual subprogram. Guarantee conditions,
in turn, describe the effect of an individual subprogram on the
environment. By constraining a subprogram by a rely condition, the global
effect of interactions is captured locally.

To make this method applicable to concrete program development and
verification tasks, its integration into tools is essential. To
capture the flexibility of the method, a number of features seem
desirable. First, we need to implement solid mathematical models for
fine-grained program behaviour. Second, we would like an abstract
layer at which inference rules and refinement laws can be derived
easily. Third, a high degree of proof automation is mandatory for the
analysis of concrete programs. In the context of the rely-guarantee
method, tools with these important features are currently missing.

This paper presents a novel approach for providing such a tool
integration in the interactive theorem proving environment
Isabelle/HOL. At the most abstract level, we use algebras to reason
about the control flow of programs as well as for deriving inference
rules and refinement laws. At the most concrete level, detailed models
of program stores support fine-grained reasoning about program data
flow and interference. These models are then linked with the
algebras. Isabelle allows us to implement these layers in a modular
way and relate them formally with one another. It not only provides us
with a high degree of confidence in the correctness of our
development, it also supports the construction of custom proof tactics
and procedures for program verification and refinement tasks.

For sequential programs, the applicability of algebra, and Kleene
algebra in particular, has been known for decades. Kleene algebra
provides operations for non-deterministic choice, sequential
composition and finite iteration, in addition to skip and abort. With
appropriate extensions, Kleene algebras support Hoare-style
verification of sequential programs, and allow the derivation of
program equivalences and refinement
rules~\cite{kozen_kleene_1997,hoare_concurrent_2011}. Kleene algebras
have been used in applications including compiler optimisation,
program construction, transformation and termination analysis, and
static analysis. Formalisations and tools are available in interactive
theorem provers such as Coq~\cite{pous_kleene_2013} and
Isabelle~\cite{armstrong_kleene_2013,armstrong_program_2013,armstrong_algebras_2014}. A
first step towards an algebraic description of rely-guarantee based
reasoning has recently been undertaken~\cite{hoare_concurrent_2011}.

The main contributions of this paper are as follows. First, we
investigate algebraic principles for rely-guarantee style
reasoning. Starting from~\cite{hoare_concurrent_2011} we extract a basic
minimal set of axioms for rely and guarantee conditions which suffice
to derive the standard rely-guarantee inference rules. These axioms
provide valuable insights into the conceptual and operational role of
these constraints. However, algebra is inherently compositional, so it
turns out that these axioms do not fully capture the semantics of
interference in execution traces. We therefore explore how the
compositionality of these axioms can be broken in the right way, so as to capture the
intended trace semantics.

Second, we link our rely-guarantee algebras with a simple trace based
semantics which so far is restricted to finite executions and
disregards termination and synchronisation. Despite the simplicity of
this model, we demonstrate and evaluate our prototypical verification
tool implemented in Isabelle by verifying two simple examples from
the literature. Beyond that our approach provides a coherent framework
from which more complex and detailed models can be implemented in the
future.

Third, we derive the usual inference rules of the rely-guarantee
method with the exception of assignment axioms directly from the
algebra, and obtain assignment axioms from our models. Our
formalisation in Isabelle allows us to reason seamlessly across these
layers, which includes the data flow and the control flow of concurrent
programs.

Taken together, our Isabelle implementation constitutes a tool
prototype for the verification and construction of concurrent
programs. We illustrate the tool with two simple examples from the
literature. The complete Isabelle code can be found
online\footnote{\url{www.dcs.shef.ac.uk/~alasdair/rg}}. A
previous Isabelle implementation of rely-guarantee reasoning is due to
Prensa Nieto~\cite{nieto_rely-guarantee_2003}. Our implementation
differs both by making the link between concrete programs and algebras
explicit, which increases modularity, and by allowing arbitrary nested parallelism.

\section{Algebraic Preliminaries}
\label{sec:KA}

Rely-guarantee algebras, which are introduced in the following
section, are based on dioids and Kleene algebras. A \emph{semiring} is
a structure $(S,+,\cdot,0,1)$ such that $(S,+,0)$ is a commutative
monoid, $(S,\cdot, 1)$ is a monoid and the distributivity laws $x\cdot
(y+z)=x\cdot z + y \cdot z$ and $(x+y)\cdot z = x\cdot z+y\cdot z$ as
well as the annihilation laws $x\cdot 0=0$ and $0\cdot x=0$ hold. A
\emph{dioid} is a semiring in which addition is idempotent:
$x+x=x$. Hence $(S,+,0)$ forms a join semilattice with least element
$0$ and partial order defined, as usual, as $x\le y\Leftrightarrow
x+y=y$. The operations of addition and multiplication are isotone with
respect to the order, that is, $x \le y $ implies $z+x\le z+y$,
$z\cdot x \le z\cdot y$ and $x\cdot z \le y\cdot z$. A dioid is
\emph{commutative} if multiplication is: $x\cdot y = y \cdot x$.

In the context of sequential programs, one typically thinks of $\cdot$
as sequential composition, $+$ as nondeterministic choice, $0$ as the
abortive action and $1$ as skip. In this context it is essential that
multiplication is not commutative. Often we use $;$ for sequential
composition when discussing programs. More formally, it is well known
that (regular) languages with language union as $+$, language
product as $\cdot$, the empty language as $0$ and the empty word
language $\{\varepsilon\}$ as $1$ form dioids. Another model is formed
by binary relations with the union of relations as $+$, the product of
relations as $\cdot$, the empty relation as $0$ and the identity
relation as $1$. A model of commutative dioids is formed by sets of
(finite) multisets or Parikh vectors with multiset addition as
multiplication.

It is well known that commutative dioids can be used for modelling the
interaction between concurrent composition and nondeterministic
choice. The following definition serves as a basis for models of
concurrency in which sequential and concurrent composition
interact.

A \emph{trioid} is a structure $(S,+,\cdot,||,0,1)$ such that
$(S,+,\cdot,0,1)$ is a dioid and $(S,+,||,0,1)$ a commutative
dioid. In a trioid there is no interaction between the sequential
composition $\cdot$ and the parallel composition $||$. On the one
hand, Gischer has shown that trioids are sound and complete for the
equational theory of series-parallel pomset languages~\cite{gischer_equational_1988},
which form a well studied model of true concurrency. On the other
hand, he has also obtained a completeness result with respect to a
notion of pomset subsumption for trioids with the additional
\emph{interchange axiom} $(w \| x) \cdot (y \| z) \le (w \cdot y) \|
(x \cdot z)$ and it is well known that this additional axiom also
holds for (regular) languages in which $||$ is interpreted as the
shuffle or interleaving operation~\cite{gischer_shuffle_1981}.

Formally, the \emph{shuffle} $\|$ of two finite words is defined
inductively as $\epsilon \| s = \{s\}$, $s \| \epsilon = \{s\}$, and
$as \| bt = a(s \| bt) \cup b(as \| t)$, which is then lifted to the
shuffle product of languages $X$ and $Y$ as $X \| Y = \{x \| y:\; x
\in X \land x \in Y\}$.

For programming, notions of iteration are essential. A \emph{Kleene
  algebra} is a dioid expanded with a star operation which satisfies
both the \emph{left unfold axiom} $1 + x\cdot x^\star\le x^\star$ and
\emph{left} and \emph{right induction axioms} $z+x \cdot y \le
y\Rightarrow x^\star\cdot z \le y$ and $z + y\cdot x \le y \Rightarrow
z\cdot x^\star \le y$. It follows that $1+x\cdot x^\star = x^\star$
and that the right unfold axiom $1+x^\star\cdot x\le x^\star$ is
derivable as well. Thus iteration $x^\ast$ is modelled as the least
fixpoint of the function $\lambda y. 1+x\cdot y$, which is the same as
the least fixpoint of $\lambda y.1+ y\cdot x$. A \emph{commutative
  Kleene algebra} is a Kleene algebra in which multiplication is
commutative.

It is well known that (regular) languages form Kleene algebras and
that (regular) sets of multisets form commutative Kleene algebras. In
fact, Kleene algebras are complete with respect to the equational
theory of regular languages as well as the equational theory of binary
relations with the reflexive transitive closure operation as the
star~\cite{kozen_completeness_1994}. Moreover, commutative Kleene
algebras are complete with respect to the equational theory of regular
languages over multisets~\cite{conway_regular_1971}. It follows that
equations in (commutative) Kleene algebras are decidable.

A \emph{bi-Kleene algebra} is a structure
$(K,+,\cdot,||,0,1,\hphantom{}^\star,\hphantom{}^{(\star)})$ where
$(K,+,\cdot,0,1,\hphantom{}^\star)$ is a Kleene algebra and
$(K,+,||,0,1,\hphantom{}^{(\star)})$ is a commutative Kleene
algebra. Bi-Kleene algebras are sound and complete with respect to the
equational theory of regular series-parallel pomset languages, and the
equational theory is again decidable~\cite{laurence_completeness_2013}. A
\emph{concurrent Kleene algebra} is a bi-Kleene algebra which
satisfies the interchange law~\cite{hoare_concurrent_2011}. It can be shown that
shuffle languages and regular series-parallel pomset languages with a
suitable notion of pomset subsumption form concurrent Kleene algebras.

In some contexts, it is also useful to add a meet operation $\sqcap$ to a bi-Kleene algebra,
such that $(K,+,\sqcap)$ is a distributive lattice. This is
particularly needed in the context of refinement, where we typically
want to represent specifications as well as programs.


\section{Generalised Hoare Logics  in Kleene Algebra}

It is well known that the inference rules of sequential Hoare logic
(except the assignment axiom) can be derived in expansions of Kleene
algebras. One approach is as
follows~\cite{moller_algebras_2006}. Suppose a suitable Boolean
algebra $B$ of \emph{tests} has been embedded into a Kleene algebra
$K$ such that $0$ and $1$ are the minimal and maximal element of $B$,
$+$ corresponds to join and $\cdot$ to meet. Complements $-$ are
defined only on $B$. Suppose further that a \emph{backward diamond
  operator} $\langle x|p$ has been defined for each $x\in K$ and $p\in
B$, which models the set of all states to which each terminating
execution of program $x$ may lead from states $q$. Finally suppose
that a \emph{forward box operator} $|x]p$ has been defined which
models the (largest) set of states from which every terminating
execution of $x$ must end in states $p$ and that boxes and diamonds
are adjoints of the Galois connection $\langle x|p \le q
\Leftrightarrow p \le |x]q$, for all $x\in K$ and $p,q\in B$. It is
then evident from the above explanations that validity of a Hoare
triple $\vdash \{p\}x\{q\}$ can be encoded as $\langle x|p \le q$ and
the weakest liberal precondition operator $\mathsf{wlp}(x,q)$ as
$|x]p$. Hence the relationship between the proof theory and the
semantics of Hoare logic is captured by the Galois connection $
\vdash\{p\}x\{q\} \Leftrightarrow p \le \mathsf{wlp}(x,q)$. It has
been shown that the relational semantics of sequential while-programs
can be encoded in these \emph{modal Kleene algebras} and that the
inference rules of Hoare logic can be
derived~\cite{moller_algebras_2006}.

In the context of concurrency, this relational approach is no longer
appropriate; the following approach by
Tarlecki~\cite{tarlecki_language_1985} can be used instead. One can
now encode validity of a Hoare triple as
\begin{equation*}
  \vdash \{x\}y\{z\} \Leftrightarrow x \cdot y \le z
\end{equation*}
for arbitrary elements of a Kleene algebra. Nevertheless all the
rules of sequential Hoare logic except the assignment axiom can still
be derived~\cite{hoare_concurrent_2011}. Tarlecki's motivating
explanations carry over to the algebraic approach.

As an example we show the derivation of a generalised while
rule. Suppose $x\cdot t\cdot y\le x$. Then $x\cdot (t\cdot y)^\ast \le
x$ by the right induction axiom of Kleene algebra and therefore
$x\cdot (t\cdot y)^\ast t'\le x\cdot t'$ for arbitrary element $t'$ by
isotonicity of multiplication. This derives the while rule
\begin{equation*}
  \frac{\vdash \{x\cdot t\}y\{x\}}{\vdash \{x\}(t\cdot y)^\ast\{t'\cdot x\}}
\end{equation*}
for a generalised while loop $(t\cdot y)^\ast\cdot t'$, which
specialises to the conventional rule when $t$ and $t'$ are, in some
sense, complements.

The correspondence to a wlp-style semantics, as in modal Kleene
algebra, now requires a generalisation of the Galois connection for
boxes and diamonds to multiplication and an upper adjoint in the form
of residuation.  This can be achieved in the context of \emph{action
  algebras}~\cite{pratt_action_1990}, which expand Kleene algebras by operations
of left and right residuation defined by the Galois connections
\begin{equation*}
  x\cdot y\le z\Leftrightarrow x\le z\leftarrow y,\qquad x\cdot y\le z \Leftrightarrow y\le x\rightarrow z.
\end{equation*}
These residuals, and now even the Kleene star, can be axiomatised
equationally in action algebras. For a comprehensive list of the
properties of action algebras and their most important models
see~\cite{armstrong_kleene_2013}, including the language and the relational
model. In analogy to the development in modal Kleene algebra we can
now stipulate $ \wlp(x,y)=y\leftarrow x$ and obtain the Galois
connection
\begin{equation*}
  \vdash \{x\}y\{z\} \Leftrightarrow x\le \wlp(y,z)
\end{equation*}
with $\vdash \{\wlp(y,z)\}y\{z\}$ and $x\le \wlp(y,z)\Rightarrow \
\vdash \{x\}y\{z\}$ as characteristic properties. Moreover, if the
action algebra is also a quantale, and infinite sums exist, it follows
that $ \wlp(y,z)=\sum\{x:\ \vdash \{x\}y\{z\}\}$.  It is obvious that
this definition makes sense in all models of action algebras and
quantales. Intuitively, suppose $p$ stands for the set of all
behaviours of a system, for instance the set of all execution traces,
that end in state $p$, and likewise for $q$. Then $\{p\}x\{q\}$ states
that all executions ending in $p$ can be extended by $x$ to executions
ending in $q$. $\wlp(x,q)$ is the most general behaviour, that is the
set of all executions $p$ after which all executions of $x$ must end
in $q$.

A residuation for concurrent composition can be considered as well:
\begin{equation*}
x\|y \le z \Leftrightarrow y \le x/z.
\end{equation*}
The residual $x/z$ represents the weakest program such that when
placed in parallel with $x$, the parallel composition behaves as $z$.


\section{A  Rely-Guarantee Algebra}
\label{sec:RG}

We now show how bi-Kleene algebras can be expanded into a
simple algebra that supports the derivation of rely-guarantee style
inference rules. This development does \emph{not} use the interchange
law for several reasons. First, this law fails for fair parallel
composition $x\parallel_f y$ in models with possibly infinite, or
non-terminating programs. In this model, $x \cdot y \not\leq x
\parallel_f y$ whenever $x$ is non-terminating. Secondly, it is not
needed for deriving the usual rules of rely-guarantee.

A rely-guarantee algebra is a structure
$(K,I,+,\sqcap,\cdot,\|,\hphantom{}^\star,0,1)$, where $(K,+,\sqcap)$
is a distributive lattice, $(K,+,\cdot,\|,0,1)$ is a trioid and
$(K,+,\sqcap,\cdot,\|,\hphantom{}^\star,0,1)$ is a bi-Kleene algebra
where we do not consider the parallel star. $I$ is a distinguished
subset of rely and guarantee conditions or \emph{interference
  constraints} which satisfy the following axioms
\begin{align}
r\|r &\le r, \label{rg1}\\
r &\le r\|r', \label{rg2}\\
r\|(x\cdot y) &= (r\|x)\cdot(r\|y), \label{rg3}\\
r\|x^+ &\le (r\|x)^+ \label{rg4}.
\end{align}
By convention, we use $r$ and $g$ to refer to elements of $I$,
depending on whether they are used as relies or guarantees, and
$x,y,z$ for arbitrary elements of $K$. The operations $\|$ and $\sqcap$
must be closed with respect to $I$.

The general idea is to constrain a program by a rely condition by
executing the two in parallel. Axiom (\ref{rg1}) states that
interference from a constraint being run twice in parallel is no
different from just the interference from that constraint begin run
once in parallel. Axiom (\ref{rg2}) states that interference from a
single constraint is less than interference from itself and another
interference constraint. Axiom (\ref{rg3}) allows an interference
constraint to be split across sequential programs. Axiom (\ref{rg4})
is similar to Axiom (\ref{rg3}) in intent, except it deals with finite
iteration.

Some elementary consequences of these rules are
\begin{equation*}
1 \le r,\qquad
r^\star = r\cdot r = r = r\|r,\qquad
r\|x^+ = (r\|x)^+.
\end{equation*}

\begin{thm}
  Axioms (\ref{rg1}), (\ref{rg2}) and (\ref{rg3}) are independent.
\end{thm}
\begin{proof}
  We have used Isabelle's
  \emph{Nitpick}~\cite{blanchette_nitpick:_2010} counterexample
  generator to construct models which violate each particular axiom
  while satisfying all others. \qed
\end{proof}

\begin{thm}
  Axiom (\ref{rg3}) implies (\ref{rg4}) in a quantale
  where $\|$ distributes over arbitrary suprema.
\end{thm}
\begin{proof}
  In a quantale $x^+$ can be defined as a sum of powers
  $x^+=\sum_{i\ge 1} x^i$ where $x^1=x$ and $x^{i+1}=x\cdot x^i$. By
  induction on $i$ we get $r\|x^i = (r\|x)^i$, hence
  \begin{align*}
    r\|x^+ = r\|\sum_{i\ge 1} x^i = \sum_{i\ge 1} r\|x^i = \sum_{i\ge 1} (r\|x)^i = (r\|x)^+.
  \end{align*}
  \qed
\end{proof}

In first-order Kleene algebras (\ref{rg3}) and (\ref{rg4}) are
independent, but it is impossible to find a counterexample with
Nitpick because it generates only finite counterexamples, and all
finite Kleene algebras are a forteriori quantales.

Jones quintuples can be encoded in this setting as
\begin{align}
r, g \vdash \{p\} x \{q\} \iff p\cdot(r\|x) \le q \land x \le g. \label{quin}
\end{align}
This means that program $x$ when constrained by a rely $r$, and
executed after $p$, behaves as $q$. Moreover, all behaviours of $x$
are included in its guarantee.

\begin{thm}
  The standard rely-guarantee inference rules can be derived with the
  above encoding, as shown in Figure \ref{fig:rgrules}.
\end{thm}

\begin{figure}[hbt]
\centering
\begin{prooftree}
\RightLabel{\ Skip}
\AxiomC{$p\cdot r \le p$}
\UnaryInfC{$r, g \vdash \{p\}1\{p\}$}
\end{prooftree}

\begin{prooftree}
\RightLabel{\ Weakening}
\AxiomC{$r' \le r$}
\AxiomC{$g \le g'$}
\AxiomC{$p \le p'$}
\AxiomC{$r', g' \vdash \{p'\}x\{q'\}$}
\AxiomC{$q' \le q$}
\QuinaryInfC{$r, g \vdash \{p\}x\{q\}$}
\end{prooftree}

\begin{prooftree}
\RightLabel{\ Sequential}
\AxiomC{$r, g \vdash \{p\}x\{q\}$}
\AxiomC{$r, g \vdash \{q\}y\{s\}$}
\BinaryInfC{$r, g \vdash \{q\}x\cdot y\{s\}$}
\end{prooftree}

\begin{prooftree}
\RightLabel{\ Parallel}
\AxiomC{$r_1, g_2 \vdash \{p_1\}x\{q_1\}$}
\AxiomC{$g_1 \le r_2$}
\AxiomC{$r_1, g_2 \vdash \{p_2\}y\{q_2\}$}
\AxiomC{$g_2 \le r_1$}
\QuaternaryInfC{$r_1 \sqcap r_2, g_1 \| g_2 \vdash \{p_1 \sqcap q_2\}x\|y\{q_1 \sqcap q_2\}$}
\end{prooftree}

\begin{prooftree}
\RightLabel{\ Choice}
\AxiomC{$r, g \vdash \{p\}x\{q\}$}
\AxiomC{$r, g \vdash \{p\}y\{q\}$}
\BinaryInfC{$r, g \vdash \{p\}x + y\{q\}$}
\end{prooftree}

\begin{prooftree}
\RightLabel{\ Star}
\AxiomC{$p\cdot r \le p$}
\AxiomC{$r, g \vdash \{p\}x\{p\}$}
\BinaryInfC{$r, g \vdash \{p\}x^\star\{p\}$}
\end{prooftree}
\caption{Rely-guarantee inference rules}
\label{fig:rgrules}
\end{figure}

Thus (\ref{rg1}) to (\ref{rg4}), which are all necessary to derive these
rules, represent a minimal set of axioms from which these
inference rules can be derived.

If we add residuals to our algebra quintuples can be encoded in the
following way, which is equivalent to the encoding in Equation
(\ref{quin}).
\begin{align}
r, g \vdash \{p\} c \{q\} \iff c \le r/(p \rightarrow q) \sqcap g \label{refine}.
\end{align}
This encoding allows us to think in terms of program refinement, as in~\cite{hayes_refining_2013}, since
$r/(p \rightarrow q) \sqcap g$ defines the weakest program that when
placed in parallel with interference from $r$, and guaranteeing
interference at most $g$, goes from $p$ to $q$---a generic
specification for a concurrent program.

\section{Breaking Compositionality}
\label{sec:INT}

While the algebra in the previous section is adequate for deriving the
standard inference rules, its equality is too strong to capture many
interesting statements about concurrent programs. Consider the
congruence rule for parallel composition, which is inherent in the
algebraic approach:
\begin{align*}
x = y \implies x\|z = y\|z.
\end{align*}
This can be read as follows; if $x$ and $y$ are equal, then they must
be equal under all possible interferences from an arbitrary $z$. At
first, this might seem to preclude any fine-grained reasoning about
interference using purely algebra. This is not the case, but breaking
such inherent compositionality in just the right way to capture
interesting properties of interference requires extra work.

The obvious way of achieving this is to expand our rely-guarantee
algebra with an additional function $\pi : K \to K$ and redefining our
quintuples as,
\begin{align*}
r, g \vdash \{p\} x \{q\} \iff p\cdot (r\|c) \le_\pi q \land x \le g.
\end{align*}
Where $x \le_\pi y$ is $\pi(x) \le \pi(y)$. Since for any operator
$\bullet$ it is not required that
\begin{align*}
\pi(x) = \pi(y) \implies \pi(x \bullet z) = \pi(y \bullet z),
\end{align*}
we can break compositionality in just the right way, provided we chose
appropriate properties for $\pi$. These properties are extracted from
properties of the trace model, which will be explained in detail in
the next section. Many of those can be derived from the fact that, in
our model, $\pi = \lambda x.\; x \sqcap c$, where $c$ is healthiness
condition which filters out ill-defined traces. We do not list these
properties here. In addition $\pi$ must satisfy the properties
\begin{align}
x^\star &\le_\pi \pi(x)^\star, \label{con1}\\
x\cdot y &\le_\pi \pi(x)\cdot \pi(y), \label{con2}\\
z + x\cdot y \le_\pi y &\implies x^\star\cdot z \le_\pi y\label{con3},\\
z + y\cdot x \le_\pi y &\implies z\cdot x^\star \le_\pi y\label{con4}.
\end{align}

For any operator $\bullet$, we write $x \bullet_\pi y$ for the
operator $\pi(x\bullet y)$, and we write $x^\pi$ for
$\pi(x^\star)$.
\begin{thm}
$(\pi(K), +_\pi, \cdot_\pi, \hphantom{}^\pi,0,1)$ is a Kleene algebra.
\end{thm}
\begin{proof}
  It can be shown that $\pi$ is a retraction, that is, $\pi^2 =
  \pi$. Therefore, $x\in \pi(K)$ iff $\pi(x) = x$. This condition can
  then be used to check the closure conditions for all operations. \qed
\end{proof}

We redefine our rely-guarantee algebra as a structure
$(K,I,+,\sqcap,\cdot,\|,\hphantom{}^\star,\pi,0,1)$ which, in addition
to the rules in Section \ref{sec:RG}, satisfies (\ref{con1}) to
(\ref{con4}).
\begin{thm}
  All rules in Figure \ref{fig:rgrules} can be derived in this
  algebra.
\end{thm}
Moreover their proofs remain the same, mutatis mutandis.

\section{Finite Language Model}
\label{sec:Model}

We now construct a finite language model satisfying the axioms in
Section \ref{sec:RG} and \ref{sec:INT}. Restricting our attention to
finite languages means we do not need to concern ourselves with
termination side-conditions, nor do we need to worry about additional
restrictions on parallel composition, e.g. fairness. However, all the
results in this section can be adapted to potentially infinite
languages, and our Isabelle/HOL formalisation already includes general
definitions by using coinductively defined lazy lists to represent
words, and having a weakly-fair shuffle operator for such infinite
languages.

We consider languages where the alphabet contains state pairs of the
form $(\sigma_1,\sigma_2) \in \Sigma^2$. A word in such a language is
\emph{consistent} if every such pair in a word has the same first
state as the previous transition's second state. For example,
$(\sigma_1,\sigma_2)(\sigma_2,\sigma_3)$ is consistent, while
$(\sigma_1,\sigma_2)(\sigma_3,\sigma_3)$ is inconsistent. Sets of
consistent words are essentially \emph{Aczel
  traces}~\cite{boer_formal_1999}, but lack the usual process
labels. We denote the set of all consistent words by $C$ and define
the function $\pi$ from the previous section as $\lambda X.\; X \cap
C$ in our model.

Sequential composition in this model is language product, as per
usual. Concurrent composition is the shuffle product defined in
Section \ref{sec:KA}. The shuffle product is associative, commutative,
and distributes over arbitrary joins. Both products share the same
unit, $\{\epsilon\}$ and zero, $\emptyset$. In Isabelle proving
properties of shuffle is surprisingly tricky (especially if one
considers infinite words). For a in-depth treatment of the shuffle
product see~\cite{mateescu_shuffle-like_1997}.

\begin{thm}
$(\mathcal{P}((\Sigma^2)^\star), \cup, \cdot,\|, \emptyset,
  \{\epsilon\})$ forms a trioid.
\end{thm}

The rely-guarantee elements in this model are sets containing all the
words which can be built from some set of state pairs in
$\Sigma^2$. We define a function $\langle R\rangle$ which
lifts a relation $R$ to a language containing words of length one for
each pair in $R$. The set of rely-guarantee conditions $I$ is then defined
as $\{r.\; \exists R. r = \langle R\rangle^\star\}$.

\begin{thm}
$(\mathcal{P}((\Sigma^2)^\star), I, \cup, \cdot,\|,
  \hphantom{}^\star,\pi,\emptyset, \{\epsilon\})$ is a rely-guarantee algebra.
\end{thm}
Since $\langle R\rangle$ is atomic, it satisfies several useful properties, such as,
\begin{align*}
\langle R\rangle^\star \| \langle S\rangle = \langle R\rangle^\star; \langle S\rangle; \langle R\rangle^\star,
\qquad \langle R\rangle^\star \| \langle S\rangle^\star = (\langle R\rangle^\star; \langle S\rangle^\star)^\star.
\end{align*}

To demonstrate how this model works, consider the graphical
representation of a language shown below.

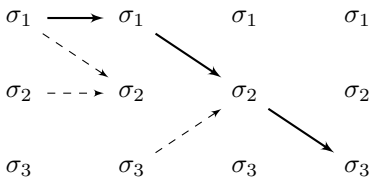
\begin{center}
\begin{tikzpicture}[x=1cm,auto]
  \node (center) {};
  \node [elem] (s12) {$\sigma_2$};
  \node [elem, above of=s12] (s11) {$\sigma_1$};
  \node [elem, below of=s12] (s13) {$\sigma_3$};
  \node [right of=s13, node distance=0.75cm] (r1) {};

  \node [elem, right of=s12, node distance=1.5cm] (s22) {$\sigma_2$};
  \node [elem, above of=s22] (s21) {$\sigma_1$};
  \node [elem, below of=s22] (s23) {$\sigma_3$};
  \node [right of=s23, node distance=0.75cm] (r2) {};

  \node [elem, right of=s22, node distance=1.5cm] (s32) {$\sigma_2$};
  \node [elem, above of=s32] (s31) {$\sigma_1$};
  \node [elem, below of=s32] (s33) {$\sigma_3$};
  \node [right of=s33, node distance=0.75cm] (r3) {};

  \node [elem, right of=s32, node distance=1.5cm] (s42) {$\sigma_2$};
  \node [elem, above of=s42] (s41) {$\sigma_1$};
  \node [elem, below of=s42] (s43) {$\sigma_3$};
  \node [right of=s43, node distance=0.75cm] (r4) {};

  \path [line] (s11) -- (s21);
  \path [line] (s21) -- (s32);
  \path [line] (s32) -- (s43);

  \path [rel] (s11) -- (s22);
  \path [rel] (s12) -- (s22);
  \path [rel] (s23) -- (s32);
\end{tikzpicture}
\end{center}
The language contains the following six words
\begin{align*}
(\sigma_1,\sigma_1)(\sigma_1,\sigma_2)(\sigma_2,\sigma_3), \qquad
(\sigma_1,\sigma_2)(\sigma_1,\sigma_2)(\sigma_2,\sigma_3),\\
(\sigma_2,\sigma_2)(\sigma_1,\sigma_2)(\sigma_2,\sigma_3), \qquad
(\sigma_1,\sigma_1)(\sigma_3,\sigma_2)(\sigma_2,\sigma_3),\\
(\sigma_1,\sigma_2)(\sigma_3,\sigma_2)(\sigma_2,\sigma_3), \qquad
(\sigma_2,\sigma_2)(\sigma_3,\sigma_2)(\sigma_2,\sigma_3),
\end{align*}
where only the first,
$(\sigma_1,\sigma_1)(\sigma_1,\sigma_2)(\sigma_2,\sigma_3)$ is
consistent. This word is highlighted with solid arrows in the
diagram above. Now if we shuffle the single state pair $(\sigma_2,
\sigma_3)$ into the above language, we might end up with a language as
below:

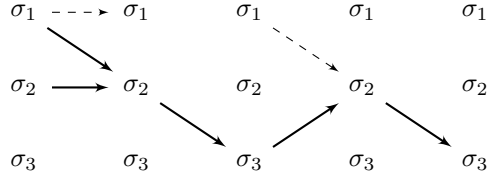
\begin{center}
\begin{tikzpicture}[x=1cm,auto]
  \node [elem] (s12) {$\sigma_2$};
  \node [elem, above of=s12] (s11) {$\sigma_1$};
  \node [elem, below of=s12] (s13) {$\sigma_3$};
  \node [right of=s13, node distance=0.75cm] (r1) {};

  \node [elem, right of=s12, node distance=1.5cm] (i2) {$\sigma_2$};
  \node [elem, above of=i2] (i1) {$\sigma_1$};
  \node [elem, below of=i2] (i3) {$\sigma_3$};
  \node [right of=i3, node distance=0.75cm] (ri) {};

  \node [elem, right of=i2, node distance=1.5cm] (s22) {$\sigma_2$};
  \node [elem, above of=s22] (s21) {$\sigma_1$};
  \node [elem, below of=s22] (s23) {$\sigma_3$};
  \node [right of=s23, node distance=0.75cm] (r2) {};

  \node [elem, right of=s22, node distance=1.5cm] (s32) {$\sigma_2$};
  \node [elem, above of=s32] (s31) {$\sigma_1$};
  \node [elem, below of=s32] (s33) {$\sigma_3$};
  \node [right of=s33, node distance=0.75cm] (r3) {};

  \node [elem, right of=s32, node distance=1.5cm] (s42) {$\sigma_2$};
  \node [elem, above of=s42] (s41) {$\sigma_1$};
  \node [elem, below of=s42] (s43) {$\sigma_3$};
  \node [right of=s43, node distance=0.75cm] (r4) {};

  \path [rel] (s11) -- (i1);
  \path [rel] (s21) -- (s32);
  \path [line] (s32) -- (s43);
  \path [line] (i2) -- (s23);

  \path [line] (s11) -- (i2);
  \path [line] (s12) -- (i2);
  \path [line] (s23) -- (s32);
\end{tikzpicture}
\end{center}

By performing this shuffle action, we no longer have a consistent word
from $\sigma_1$ to $\sigma_3$, but instead a consistent word from
$\sigma_2$ to $\sigma_3$ and $\sigma_1$ to $\sigma_3$. These new
consistent words were constructed from previously inconsistent
words---the shuffle operator can generate many consistent words from
two inconsistent words. If we only considered consistent words, \`a la
Aczel traces, we would be unable to define such a shuffle operator
directly on the traces themselves, and would instead have to rely on
some operational semantics to generate traces.

\section{Enriching the Model}

To model and verify programs we need additional concepts such as tests
and assignment axioms. A \emph{test} is any language $P$ where $P \le
\langle \Id\rangle$. We write $\test(P)$ for $\langle
\Id_P\rangle$. In Kleene algebra the sequential composition of two
tests should be equal to their intersection. However, the traces
$\test(P); \test(Q)$ and $\test(P \cap Q)$ are incomparable, as all
words in the former have length two, while all the words in the latter
have length one. To overcome this problem, we use the concepts of
\emph{stuttering} and \emph{mumbling},
following~\cite{brookes_full_1993}
and~\cite{dingel_refinement_2002}. We inductively generate the
\emph{mumble language} $w^\dagger$ for a word $w$ in a language over
$\Sigma^2$ as follows: Assume $\sigma_1,\sigma_2,\sigma_3 \in \Sigma$
and $u,v,w \in (\Sigma^2)^\star$. First, $w \in w^\dagger$. Secondly,
if $u(\sigma_1,\sigma_2)(\sigma_2,\sigma_3)v \in w^\dagger$ then
$u(\sigma_1,\sigma_3)v \in w^\dagger$. This operation is lifted to
languages in the obvious way as
\begin{align*}
X^\dagger = \bigcup\{x^\dagger.\; x \in X\}.
\end{align*}
Stuttering is represented as a rely condition $\langle \Id\rangle^\star$
where $\Id$ is the identity relation. Two languages $X$ and $Y$ are
equal under stuttering if $\langle \Id\rangle^\star \| X =_\pi \langle
\Id\rangle^\star \| Y$.

With mumbling we now have that
\begin{align*}
\test(P \cap Q) \le_\pi \test(P); \test(Q)
\end{align*}
as the longer words in $\test(P);\test(Q)$ can be mumbled down into
the shorter words of $\test(P\cap Q)$, whereas stuttering gives us the opposite direction,
\begin{align*}
\langle \Id\rangle^\star \| (\test(P);\test(Q)) \le_\pi \langle \Id\rangle^\star \| \test(P \cap Q).
\end{align*}
We henceforth assume that all languages are
implicitly mumble closed.

Using tests, we can encode if statements and while loops
\begin{align*}
\mathsf{if }\;\; P \text{ \{ } X \text { \} \textsf{else} \{ } Y \text{ \}} &= \test(P); X + \test(- P); Y,\\
\mathsf{while}\;\; P \text{ \{ } X \text{ \}} &= (\test(P);X)^\star;\test(- P).
\end{align*}
Next, we define the operator $\edn(P)$ which contains all the
words which end in a state satisfying $P$. Some useful properties of $\edn$ include
\begin{gather*}
\edn(P); \test{Q} \le_\pi \edn(P \cap Q), \qquad
\test(P) \le \edn(Q),\\
\text{range}(\Id_P \circ R) \le P \implies \edn(P); \langle R\rangle^\star \le_\pi \edn(P).
\end{gather*}

In this model, assignment is defined as
\begin{align*}
  x := e \ =\ \bigcup v.\,\test\{\sigma.\,\eval(\sigma,e) = v\} \cdot x \leftarrow v
\end{align*}
where $x \leftarrow v$ denotes the atomic command which assigns the
value $v$ to $x$. The $\eval$ function evaluates an expression $e$ in
the state $\sigma$. Using this definition we derive the assignment
rule
\begin{align*}
&\unchanged(\mathsf{vars}(e)) \cap \preserves(P) \cap \preserves(P[x/e]),\\
&\unchanged(- \{x\})\\
&\vdash \{\edn(P)\}\ x := e\ \{\edn(P[x/e])\}.
\end{align*}

The rely condition states the following: First, the environment is not
allowed to modify any of the variables used when evaluating $e$,
i.e. those variables must remain unchanged. Second, the
environment must preserve the precondition. Third, the postcondition
of the assignment statement is also preserved. In turn, the assignment
statement itself guarantees that it leaves every variable other than
$x$ unchanged. Preserves and unchanged are defined as
\begin{align*}
\preserves(P) &= \langle \{(\sigma,\sigma').\; P(\sigma) \implies P(\sigma')\}\rangle^\star,\\
\unchanged(X) &= \langle \{(\sigma,\sigma').\; \forall v\in X.\; \sigma(v) = \sigma'(v)\}\rangle^\star.
\end{align*}
We also defined two futher rely conditions, increasing and decreasing,
which are defined much like unchanged except they only require that
variables increase or decrease, rather than stay the same. We can
easily define other useful assignment rules---if we know properties about
$P$ and $e$, we can make stronger guarantees about what $x := e$ can
do. For example the assignment $x := x - 2$ can also guarantee that
$x$ will always decrease.

\section{Examples}
\label{sec:Examples}

To demonstrate how the parallel rule behaves, consider the following
simple statement, which simply assigns two variables in parallel:
\begin{align*}
&\langle\Id\rangle^\star, \langle\top\rangle^\star \vdash \{\edn(x = 2 \land y = 2 \land z = 5)\}\\
&\qquad x := x + 2 \;\|\; y := z\\
&\{\edn (x = 4 \land y = 5 \land z = 5)\}.
\end{align*}
The environment $\langle\Id\rangle^\star$ is only giving us stuttering
interference. Since we are considering this program in isolation, we
make no guarantees about how this affects the environment. To apply
the parallel rule from Figure \ref{fig:rgrules}, we weaken or
strengthen the interference constrains and pre/postcondition as needed
to fit the form of the parallel rule.

First, we weaken the rely condition to $\unchanged\{x\}\ \sqcap\ \unchanged\{y,z\}$. Second we strengthen the guarantee condition to
$\unchanged\{y,z\}\ \|\ \unchanged\{x\}$. When we apply the parallel
rule each assignment's rely will become the other assignment's
guarantee. Finally, we split the precondition and postcondition into
$\edn(x = 2)\ \sqcap\ \edn(y = 2 \land z = 5)$ and $\edn(x =
4)\ \sqcap\ \edn(y = 5 \land z = 5)$ respectively. Upon applying the
parallel rule, we obtain two trivial goals
\begin{align*}
&\langle\unchanged\{x\}\rangle^\star, \langle\unchanged\{y,z\}\rangle^\star \vdash \{\edn(x = 2)\}\;x := x + 2 \;\{\edn (x = 4)\},\\
&\langle\unchanged\{y,z\}\rangle^\star, \langle\unchanged\{x\}\rangle^\star \vdash \{\edn(y = 2 \land z = 5)\}\\
&\qquad y := z \\
&\{\edn (y = 5 \land z = 5)\}.
\end{align*}

Figure \ref{fig:findp} shows the FINDP program, which has been
used by numerous authors
e.g.~\cite{owicki_axiomatic_1975,jones_development_1981,de_roever_concurrency_2001,hayes_refining_2013}. The
program finds the least element of an array satisfying a predicate
$P$. The index of the first element satisfying $p$ is placed in the
variable $f$. If no element of the array satisfies $P$
then $f$ will be set to the length of the array. The program has two
subprograms, $A$ and $B$, running in parallel, one of which searches
the even indices while the other searches the odd indices. A speedup
over a sequential implementation is achieved as $A$ will terminate
when $B$ finds an element of the array satisfying $P$ which is less
than $i_A$.
\begin{figure}
\[
\begin{aligned}
&f_A := \mathsf{len}(\mathsf{array});\\
&f_B := \mathsf{len}(\mathsf{array});\\
&\left(
\begin{aligned}
&i_A = 0\\
&\text{\textsf{while} } i_A < f_A \land i_A < f_B \text{ \{ }\\
&\qquad\text{\textsf{if} } P(\mathsf{array}[i_A]) \text{ \{}\\
&\qquad\qquad f_A := i_A\\
&\qquad\text{\} \textsf{else} \{}\\
&\qquad\qquad i_A := i_A + 2\\
&\qquad\text{\}}\\
&\text{\}}
\end{aligned}
\middle\|
\begin{aligned}
&i_B = 1\\
&\text{\textsf{while} } i_B < f_A \land i_B < f_B \text{ \{}\\
&\qquad\text{\textsf{if} } P(\mathsf{array}[i_B]) \text{ \{}\\
&\qquad\qquad f_B := i_B\\
&\qquad\text{\} \textsf{else} \{}\\
&\qquad\qquad i_B := i_B + 2\\
&\qquad\text{\}}\\
&\text{\}}
\end{aligned}
\right);\\
&f = \text{min}(f_A,f_B)
\end{aligned}
\]
\caption{FINDP Program}
\label{fig:findp}
\end{figure}

Here, we only sketch the correctness proof, and comment on its
implementation in Isabelle. We do not attempt to give a detailed
proof, as this has been done many times previously.

To prove the correctness of FINDP, we must show that
\begin{align*}
\text{FINDP} \le_\pi \edn(\mathsf{leastP}(f)) + \edn(f = \mathsf{len}(\text{array})),
\end{align*}
where $\mathsf{leastP(f)}$ is the set of states where $f$ is the least
index satisfying $P$, and $f = len(array)$ is the set of states where
$f$ is the length of the array. In other words, either we find the
least element, or no elements in the array satisfy $P$, in which case
$f$ remains the same as the length of the array.

To prove the parallel part of the program, subprogram $A$ guarantees
that it does not modify any of the variables used by subprogram $B$,
except for $f_A$, which it guarantees will only ever
decrease. Subprogram $B$ makes effectively the same guarantee to
$A$. Under these interference constraints we then prove that $A$ or
$B$ will find the lowest even or odd index which satisfies $P$
respectively---or they do not find it, in which case $f_A$ or $f_B$
will remain equal to the length of the array.

Despite the seemingly straightforward nature of this proof, it turns
out to be surprisingly difficult in Isabelle. Each atomic step
needs to be shown to satisfy the guarantee of its containing
subprogram, as well as any goals relating to its pre and post
conditions. This invariably leads to a proliferation of many small
proof goals, even for such a simple program. More work must be done
to manage the complexity of such proofs within interactive theorem
provers.

\section{Conclusion}

We have introduced variants of semirings and Kleene algebras intended
to model rely-guarantee and interference based reasoning. We have
developed an interleaving model for these algebras which uses familiar
concepts from traces and language theory. This theory has been
implemented in the Isabelle/HOL theorem prover, providing a solid
mathematical basis on which to build a tool for mechanised refinement
and verification tasks. In line with this aim, we have applied
our formalisation to two simple examples.

This implementation serves as a basis from which further interesting
aspects of concurrent programs, such as non-termination and
synchronisation can be explored. As mentioned in Section
\ref{sec:Model}, some of the work needed to implement this we have
already done in Isabelle.

Algebra plays an important role in our development. First, it
allowed us to derive inference rules rapidly and with little proof
effort. Second, it yields an abstract layer at which many properties
that would be difficult to prove in concrete models can be verified
with relative ease by equational reasoning. Third, as pointed out in
Section \ref{sec:KA}, some fragments of the algebras considered are
decidable. Therefore, decision procedures for some aspects of
rely-guarantee reasoning can be implemented in interactive theorem
proving tools such as Isabelle.

The examples from Section \ref{sec:Examples} confirm previous
evidence~\cite{nieto_rely-guarantee_2003} that even seemingly
straightforward concurrency verification tasks can be tedious and complex. It is
too early to draw informed conclusions, but while part of this complexity
may be unavoidable, more advanced models and proof automation are
needed to overcome such difficulties. Existing work on combining
rely-guarantee with separation logic~\cite{vafeiadis_modular_2008} may
prove useful here. Our language model is sufficiently generic such
that arbitrary models of stores may be used, including those common in
separation logic, which have already been implemented in Isabelle~\cite{klein_mechanised_2012}.

In addition, algebraic approaches to separation logic have already
been introduced. Examples are the separation algebras
in~\cite{calcagno_local_2007}, and algebraic separation
logic~\cite{dang_algebraic_2011}. More recently, concurrent Kleene
algebras have given an algebraic account of some aspects of concurrent
separation logic~\cite{hoare_concurrent_2011,hoare_locality_2011}.

\paragraph*{Acknowledgements.} The authors would like to thank Brijesh
Dongol and Ian Hayes for inspiring discussions on concurrency
verification and the rely-guarantee method. The first author
acknowledges funding from an EPSRC doctoral fellowship. The second
author is supported by CNPq Brazil. The third
author acknowledges funding by EPSRC grant EP/J003727/1.

\bibliography{paper}{}
\bibliographystyle{plain}

\end{document}